\documentclass[pra,aps,twocolumn,superscriptaddress]{revtex4-1}
\usepackage[percent]{overpic}
\usepackage{mathtools}
\usepackage[british]{babel}
\usepackage{amsmath,amssymb,amsthm,dsfont,hyperref,ctable,url,commath}
\usepackage{color}
\usepackage{bbold}
\usepackage{braket}
\usepackage{units}
\input{qcircuit}

\newcommand{\Tr}{\operatorname{Tr}}
\newcommand{\supp}{\operatorname{supp}}
\newcommand{\spn}{\operatorname{span}}

\newcommand{\hilb}[1]{\mathcal{#1}}

\newtheorem{thm}{Theorem}
\newtheorem{lmm}{Lemma}
\newtheorem{cor}{Corollary}
\newtheorem{dfn}{Definition}

\newcommand{\cc}[3]{\mathcal{P}^{{#1}\to{#2}}_{#3}}
\newcommand{\cs}[1]{C_{#1}}
\newcommand{\qs}[1]{Q_{#1}}

\usepackage{enumitem}

\usepackage{array}
\newcolumntype{?}{!{\vrule width 1.2pt}}

\usepackage{soul}

\definecolor{blue1}{rgb}{0.03, 0.27, 0.49}
\begin{document}

\title{No-hypersignaling principle}

\author{Michele \surname{Dall'Arno}}

\email{cqtmda@nus.edu.sg}

\affiliation{Centre for  Quantum Technologies, National
  University of  Singapore, 3 Science Drive  2, 117543,
  Singapore}

\author{Sarah \surname{Brandsen}}

\email{sbrandse@caltech.edu}

\affiliation{Centre for  Quantum Technologies, National
  University of  Singapore, 3 Science Drive  2, 117543,
  Singapore}

\author{Alessandro \surname{Tosini}}

\email{alessandro.tosini@unipv.it}

\affiliation{QUIT    group,   Physics    Dept.,   Pavia
  University, and  INFN Sezione di Pavia,  via Bassi 6,
  27100 Pavia, Italy}

\author{Francesco \surname{Buscemi}}

\email{buscemi@is.nagoya-u.ac.jp}

\affiliation{Graduate  School  of  Informatics,  Nagoya
  University, Chikusa-ku, 464-8601 Nagoya, Japan}

\author{Vlatko \surname{Vedral}}

\email{phyvv@nus.edu.sg}

\affiliation{Atomic   and   Laser  Physics,   Clarendon
  Laboratory, University of  Oxford, Parks Road, Oxford
  OX13PU, United Kingdom}

\affiliation{Centre for  Quantum Technologies, National
  University of  Singapore, 3 Science Drive  2, 117543,
  Singapore}

\begin{abstract}
  A paramount  topic in quantum foundations,  rooted in
  the study  of the EPR paradox  and Bell inequalities,
  is that of characterizing  quantum theory in terms of
  the space-like  correlations it allows. Here  we show
  that to focus only  on space-like correlations is not
  enough: we  explicitly construct  a toy  model theory
  that, while  not contradicting classical  and quantum
  theories  at the  level  of space-like  correlations,
  still displays an anomalous behavior in its time-like
  correlations.   We call  this anomaly,  quantified in
  terms   of  a   specific   communication  game,   the
  ``hypersignaling'' phenomena.  We hence conclude that
  the  ``principle  of  quantumness,''  if  it  exists,
  cannot  be found  in  space-like correlations  alone:
  nontrivial  constraints need  to be  imposed also  on
  time-like   correlations,   in   order   to   exclude
  hypersignaling theories.
\end{abstract}

\maketitle

One of the main tenets in modern physics is that if two
space-like separated  events are correlated,  then such
correlations       must       not       carry       any
information~\cite{Ein05}. This assumption, constituting
the so-called \textit{no-signaling  principle}, was the
starting  point used  by Bell~\cite{Bel64}  to quantify
and  compare   space-like  correlations   of  different
theories on even grounds---an  idea of vital importance
for his argument about the EPR paradox~\cite{EPR35} and
the derivation of  his famous inequality. Subsequently,
due     to      seminal     works      by     Tsirelson
(Cirel'son)~\cite{Tsi80}      and      Popescu      and
Rohrlich~\cite{PR94},   it   became  clear   that   the
no-signaling   principle  alone   is   not  enough   to
characterize   ``physical''  space-like   correlations:
non-signaling   space-like   correlations  allowed   by
quantum theory form a \textit{strict} subset within the
set of all non-signaling correlations~\cite{Pop2014}.

A  natural   question  is  then  to   try  to  identify
additional   principles   that,   together   with   the
no-signaling  principle, may  be able  to rule  out all
super-quantum   non-signaling  correlations   at  once.
Various   ideas  have   been  proposed,   ranging  from
complexity theory, e.g.  the collapse of the complexity
tower~\cite{vanDam05} to information  theory, e.g.  the
information     causality    principle~\cite{PPKSWZ09}.
However, none  of these  has been able  to characterize
the   quantum/super-quantum  boundary   in  full.    In
particular,  an outstanding  open  question is  whether
quantum  theory can  be characterized  in terms  of the
space-like correlations it allows~\cite{Pop2014}.

In this  paper, we show  that this cannot be  done: any
approach to  characterize quantum theory based  only on
space-like   correlations  is   necessarily  incomplete
unless   it   also   takes   into   account   time-like
correlations   as  well.    Our   approach,  which   is
completely   unrelated  to   the   study  of   temporal
correlations               \textit{\`a              la}
Leggett--Garg~\cite{BKMPP15, BE14,  BMKG13, MKTLSPK14},
considers   the   elementary  resource   of   noiseless
communication  and the  input/output correlations  that
can   be  so   established.    By   analogy  with   the
no-signaling principle, we operationally introduce what
we  call  the  ``no-hypersignaling  principle,''  which
roughly states  that any input/output  correlation that
can  be obtained  by  transmitting  a composite  system
should also be obtainable by independently transmitting
its constituents.  As  obvious as this may  look (it is
indeed so in classical  and quantum theories), the fact
that   quantum  theory   obeys  the   no-hypersignaling
principle  (as  we  define  it) is  in  fact  a  highly
nontrivial consequence  of a  recent result  by Frenkel
and  Weiner~\cite{FW15}.   We   also  notice  that  the
no-hypersignaling   principle  is   not  related   with
phenomena  such  as  superadditivity of  capacities  of
noisy quantum channels~\cite{Hast09}.

We then  construct a  toy model theory,  which violates
the  no-hypersignaling  principle, but  only  possesses
classical  space-like  correlations.    As  such,  this
theory   (and  other   analogous  theories)   would  go
undetected  in  any   test  involving  only  space-like
correlations, despite  displaying the  anomalous effect
of hypersignaling.  On the technical side, our model is
closely       related       to       the       standard
implementation~\cite{Bar05,     Bar07,     DT10}     of
Popescu--Rohrlich~\cite{PR94}             super-quantum
non-signaling space-like correlations (or ``PR-boxes,''
for  short).  However,  while the  PR-box model  theory
relies  on  entangled   states  to  outperform  quantum
\textit{space}-like  correlations,  our  hypersignaling
model  relies  on  \textit{entangled  measurements}  to
outperform  quantum   \textit{time}-like  correlations.
Nonetheless, since  in our model only  separable states
are available, no  super-quantum space-like correlation
can be obtained.  Therefore,  while the standard PR-box
model  theory can  be ruled  out  on the  basis of  its
super-quantum   space-like   correlations,  the   model
proposed here can only be ruled out by the principle of
no-hypersignaling.

\emph{The   No-Hypersignaling   Principle}.    ---   In
general, the starting point of  a physical theory is to
define its elementary  systems.  In \textit{generalized
  probabilistic     theories}     (see     Supplemental
Material~\cite{SM},   and  Refs.\cite{S07,   CDP11})  a
system  $S =  (\mathcal{S}, \mathcal{E})$  is typically
defined by giving  a set of states  $\mathcal{S}$ and a
set of effects $\mathcal{E}$, representing respectively
the preparations  and the  observations of  the system.
States    can   be    arranged   to    form   ensembles
$\{\Omega_0,  \Omega_1,\cdots \}$  and  effects can  be
arranged  to  form  measurements  $\{E_0,E_1,\cdots\}$.
The theory must also comprise  a rule for computing the
conditional  probability of  any effect  on any  state.
For example, in quantum  theory, a system is associated
with a $d$-dimensional Hilbert space $\hilb{H}$, states
and effects  are represented by  positive semi-definite
operators on $\hilb{H}$,  and conditional probabilities
are given  by the Born  (trace) rule.  The  theory must
also include  a set  of transformations  mapping states
into states (or  effects into effects): in  the case of
quantum  theory, this  is the  set of  quantum channels
(i.e.  completely positive  and trace preserving linear
maps).

Given an elementary system, an important role is played
by   its  \textit{dimension}~\cite{BKLS14},   which  is
expected  to  depend  solely   on  the  set  of  states
$\mathcal{S}$  and  effects $\mathcal{E}$.   Since  one
usually  assumes that  convex  mixtures  of states  and
effects can  always be  considered (following  the idea
that the randomization of different experimental setups
is  in  itself  another valid  experiment),  by  linear
extension it  is natural  to introduce the  real vector
spaces          $\mathcal{S}_{\mathbb{R}}$          and
$\mathcal{E}_{\mathbb{R}}$,  generated  by real  linear
combinations   of   elements   of   $\mathcal{S}$   and
$\mathcal{E}$,  respectively.  Notice  that in  typical
situations  $\mathcal{E}_{\mathbb{R}}$  coincides  with
the      set     of      linear     functionals      on
$\mathcal{S}_{\mathbb{R}}$.   One soon  arrives at  the
following definition:

\begin{dfn}[Linear dimension]
  The  linear dimension  of  a system  $S$, denoted  by
  $\ell(S)$, is  defined as  the dimension of  the real
  vector    space     $\mathcal{S}_{\mathbb{R}}$    (or
  $\mathcal{E}_{\mathbb{R}}$, which is  the same in the
  finite dimensional case considered in this work).
\end{dfn}

The  linear dimension  of a  classical system  with $d$
extremal  states is  equal to  $d$, whereas  the linear
dimension  of  a  quantum   system  associated  with  a
$d$-dimensional   Hilbert    space   is    $d^2$.   For
convenience,  we  denote  a  $d$-dimensional  classical
system   by  $\cs{d}$   and  a   quantum  system   with
$d$-dimensional Hilbert  space by $\qs{d}$ so  that, in
formula, $\ell(\cs{d})=d$ and $\ell(\qs{d})=d^2$.

There are various  ways proposed to make  sense of this
discrepancy:  a  typical  solution   is  to  define  an
``operational''  dimension  as  the maximum  number  of
states   that  can   be  distinguished   in  a   single
measurement, see, e.g.,  Ref.~\cite{HW12}. In this way,
even though the linear dimension of a quantum system is
$d^2$, the  operational dimension turns out  to be $d$,
thus matching  the dimension of the  underlying Hilbert
space.  In  what follows,  we introduce  an alternative
operational  definition  of  dimension  which  is  both
widely applicable and is independent of any arbitrarily
chosen task, such as perfect state discrimination.

In order to make our  analysis more concrete we need to
introduce  some notation.   Given two  finite alphabets
$\mathcal{X}=\{x\}$ and  $\mathcal{Y}=\{y\}$ containing
$m$ and $n$ letters,  respectively, let us consider the
set of all $m$-input/$n$-output conditional probability
distributions  $p_{y|x}$  that   can  be  generated  by
transmitting  one  elementary  system  $S$,  when  free
shared  randomness  between   sender  and  receiver  is
allowed. With this,  we mean that the input  $x$ can be
``encoded''          on          some          ensemble
$\{\Omega_x^{(\lambda)}:x\in\mathcal{X}  \}$ while  the
output   letter  $y$   is   ``decoded''  whenever   the
corresponding outcome  is obtained in  some measurement
$\{E_y^{(\lambda)}:y\in\mathcal{Y}\}$,  where $\lambda$
parameterizes  the shared  random variable.   We denote
the   convex   set   of  all   such   correlations   by
$\cc{m}{n}{S}$.   For  example, $\cc{m}{n}{\cs{d}}$  is
the   set  of   all  $m$-input/$n$-output   conditional
probability distributions that can be obtained by means
of  a $d$-dimensional  classical noiseless  channel and
shared random  data.  Equivalently, $\cc{m}{n}{\cs{d}}$
can be characterized as the polytope whose vertices are
exactly all  those $p_{y|x}$  with either null  or unit
entries and  such that $p_y  := \sum_x p_{y|x}  \neq 0$
for at most $d$ different values of $x$.

Crucial in our  analysis is a recent  result by Frenkel
and Weiner~\cite{FW15},  stating that, in  the presence
of  shared   classical  randomness,   any  input/output
correlation obtainable  with a  $d$-dimensional quantum
system  is  also   obtainable  with  a  $d$-dimensional
classical system (and vice versa)---in formula,
\begin{align*}
  \cc{m}{n}{\cs{d}}=\cc{m}{n}{\qs{d}},
\end{align*}
for all  (finite) values of  $m$ and $n$.  We  are thus
motivated to introduce the following definition:
\begin{dfn}[Signaling dimension]
  The signaling  dimension of a system  $S$, denoted by
  $\kappa(S)$, is  defined as the smallest  integer $d$
  such  that  $\cc{m}{n}{S}\subseteq\cc{m}{n}{\cs{d}}$,
  for all $m$ and $n$.
\end{dfn}
Note that $\kappa(S)$ equals  the usual dimension, both
in  classical  and  quantum  theories, and  is  thus  a
natural  candidate  for  an operational  definition  of
dimension.  Moreover,  $\kappa(S)$ only depends  on the
structure of  $\mathcal{S}$ and  $\mathcal{E}$, without
relying  on  the  (arbitrarily   made)  choice  of  any
specific protocol  such as state  discrimination. Also,
due to the already  mentioned result of~\cite{FW15}, in
what   follows   we   will  simply   use   the   symbol
$\cc{m}{n}{d}$ to denote $\cc{m}{n}{\cs{d}}$, since the
fact that the underlying theory is classical or quantum
is immaterial for the problem at hand.

\begin{figure}[t!]
  \begin{overpic}[width=.7\columnwidth]{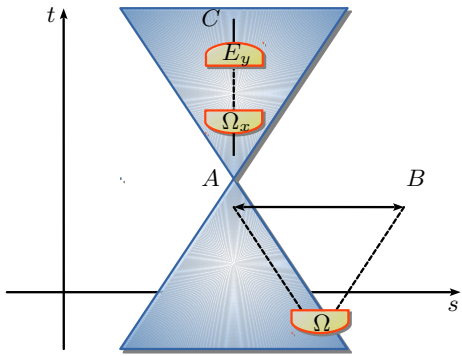}
    \put (43, 37) {$A$}
    \put (88, 37) {$B$}
    \put (43, 72) {$C$}
    \put (67.5, 5.1) {$\Omega$}
    \put (47.5, 49.5) {$\Omega_x$}
    \put (47.5, 64.5) {$E_y$}
    \put (97, 9.5) {$s$}
    \put (9, 73) {$t$}
  \end{overpic}
  \caption{\textbf{Space-like       and       time-like
      correlations.}  Events $A$ and $B$ are space-like
    separated, i.e.  information cannot travel from the
    one   to  the   other  (no-signalling   principle).
    Correspondingly,  they  can only  share  space-like
    correlations, previously distributed in the form of
    a bipartite state $\Omega$.  Events $A$ and $C$ are
    time-like  separated, i.e.  information can  indeed
    travel from $A$ to $C$: such information is encoded
    into the states $\{ \Omega_x \}$, and later decoded
    by   the  measurement   $\{   E_y   \}$.   As   the
    no-signalling   principle   constrains   space-like
    correlations,   the   no-hypersignaling   principle
    constrains time-like correlations.}
  \label{fig:minkowski}
\end{figure}

The  no-hypersignaling   principle  is   introduced  by
looking at how the  dimension behaves under composition
of elementary systems. In order to do this, we need the
theory to provide us with a rule for combining multiple
elementary systems into a  larger one.  For example, in
quantum theory,  the composition  rule is given  by the
tensor product  of the  underlying Hilbert  spaces. For
the  sake of  the  present  paper, we  do  not need  to
understand the  various possible mechanisms  with which
elementary  systems can  be  composed: given  a set  of
elementary   systems   $\{S_k\}$,   we   denote   their
composition by $\otimes_k S_k$.  Notice that the tensor
product  should here  be interpreted  only as  a symbol
denoting  composition, and  is not  necessarily related
with the  actual tensor  product of vector  spaces (the
interested reader may refer to Ref.~\cite{BBLW08}).

However, it  is natural to assume  that the composition
rule  must  satisfy  some  sensible  constraints.   For
example,  a first  condition that  must be  met by  any
self-consistent theory is that  any circuit obtained as
the composition of systems,  states, effects, and their
transformations should produce non-negative conditional
probabilities.    An  additional   condition  is   that
space-like    correlations   obey    the   no-signaling
principle,  so  that   any  instantaneous  exchange  of
information is forbidden, see Fig.~\ref{fig:minkowski}.
There are still other,  more subtle conditions that can
be considered.

For example,  Ref.~\cite{HW12} considers  the condition
of \textit{local tomography}.   This requires the state
of  a   composite  system  to  be   determined  by  the
statistics  of measurements  done independently  on its
constituents.  This principle is not as obvious as that
of  no-signaling,   however,  it  arguably   remains  a
sensible requirement for a theory that does not want to
be ``too holistic'' (namely, the state of any composite
systems  should  always  be locally  accessible).   The
principle  of  local  tomography is  related  with  the
notion of  dimension: a  theory is  locally tomographic
whenever  the linear  dimension of  a composite  system
does not exceed the product of the linear dimensions of
its constituents, in formula,
\begin{equation}
  \label{eq:local-tomo}
  \ell(\otimes_kS_k)\le\prod_k\ell(S_k).
\end{equation}
In fact, without resorting  to truly exotic, \textit{ad
  hoc} theories, the linear  dimension of the composite
system cannot be strictly less  than the product of the
linear   dimensions  of   its   constituents,  so   the
inequality in  Eq.~(\ref{eq:local-tomo}) can  be safely
replaced with the equal  sign (see Ref.~\cite{HW12} for
further details on the concept of linear dimension).

The  no-hypersignaling  principle  is the  analogue  of
Eq.~(\ref{eq:local-tomo})  stated   for  the  signaling
dimension, rather  than the  linear dimension.  We thus
have the following definition:

\begin{dfn}[No-hypersignaling principle]
  A theory  is non-hypersignaling  if and only  if, for
  any  set   of  systems  $\{S_k  \}$   with  signaling
  dimensions $\kappa(S_k)$, the  signaling dimension of
  the composite system $\otimes_k S_k$ satisfies
  \begin{align}
    \kappa(\otimes_k S_k)\le \prod_k \kappa(S_k).
  \end{align}
\end{dfn}

In particular, the no-hypersignaling principle requires
that,  given two  copies of  the same  system $S$  with
signaling  dimension $d$,  the  signaling dimension  of
$S\otimes S$ cannot exceed $d^2$, in formula
\begin{align*}
  \cc{m}{n}{S}    \subseteq    \cc{m}{n}{d}    \implies
  \cc{m}{n}{S^{\otimes 2}} \subseteq \cc{m}{n}{d^2},
\end{align*}
for  all $m$  and $n$.   The situation  is depicted  in
Fig.~\ref{fig:dimadditivity}.

\begin{figure}[t!]
  \begin{overpic}[width=.6\columnwidth]{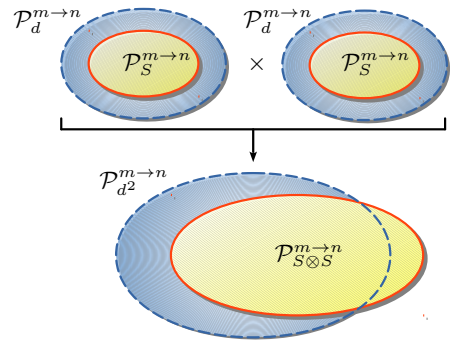}
    \put (16, 69) {$\cc{m}{n}{S}$}
    \put (72, 69) {$\cc{m}{n}{S}$}
    \put (48, 69) {$\times$}
    \put (-12, 80) {$\cc{m}{n}{d}$}
    \put (47, 80) {$\cc{m}{n}{d}$}
    \put (10, 39) {$\cc{m}{n}{d^2}$}
    \put (55, 20) {$\cc{m}{n}{S \otimes S}$}
  \end{overpic}
  \caption{{\bf   Illustration   of  a   hypersignaling
      theory}.   While the  system $S$  alone satisfies
    $\cc{m}{n}{S} \subseteq \cc{m}{n}{d}$, and thus has
    signaling  dimension $d$,  the composite  system $S
    \otimes  S$  has  a  signaling  dimension  strictly
    larger than $d^2$.}
  \label{fig:dimadditivity}
\end{figure}

Roughly  speaking,  while  the  no-signaling  principle
prevents  \textit{space}-like  separated  parties  from
communicating, the no-hypersignaling principle prevents
\textit{time}-like separated parties from communicating
``too much,''  see again  Fig.~\ref{fig:minkowski}.  It
may help to think  that the no-hypersignaling principle
guarantees   that    the   input/output   correlations,
attainable when transmitting two elementary systems, do
not  depend   on  whether  the  systems   are  actually
transmitted in series or in parallel.

Before proceeding,  it will be useful  to interpret the
no-hypersignaling principle in terms of a communication
game.  To this aim, let us denote a composite system by
$\bar   S=\otimes_kS_k$   and   by  $K$   the   product
$\prod_k\kappa(S_k)$ of the local signaling dimensions.
It is  therefore a  straightforward application  of the
hyperplane   separation  theorem   that  a   theory  is
hypersignaling if  and only if,  for some $m$  and $n$,
there exists a  conditional probability distribution $p
\in \cc{m}{n}{\bar  S}$ and an $m\times  n$ real matrix
$g$, such that
\begin{align}
  \label{eq:iff}
  g^T \cdot p > \max_{q \in \cc{m}{n}{K}} g^T \cdot q,
\end{align}
where we use the notation $A^T \cdot B$ to indicate the
Hilbert-Schmidt    dot-product   $\sum_{x,y}    A_{x,y}
B_{x,y}=\Tr[A^TB]$.    Notice  that   the  maximization
problem  in the  r.h.s.   of  Eq.~\eqref{eq:iff} is  in
closed form:  by linearity  the maximum is  attained on
the vertices of the  polytope $\cc{m}{n}{K}$, which are
finite  in  number  and computed  in  the  Supplemental
Material~\cite{SM}.

The  matrix  $g$  can  be  interpreted  as  the  payoff
function  defining  a  communication  game,  where  the
sender inputs $x$ and the receiver outputs $y$, leading
to  the  corresponding  payoff  $g_{x,y}$.   From  this
viewpoint, Eq.~\eqref{eq:iff} represents the fact that,
for any game  $g$, the average payoff  of the composite
system $\bar S$ never exceeds the payoff of the product
of  its  parts  $\{S_k\}$.    A  general  framework  to
consider   such    game-theoretic   interpretation   is
developed  in the  Supplemental Material~\cite{SM},  by
extending    the    theory    of    extremal    quantum
measurements~\cite{Par99,     DLP05}     to     general
probabilistic theory.

\emph{The  counterexample}.  ---  In  what follows,  we
exploit our general framework  to construct a toy model
theory that  violates the  no-hypersignaling principle,
namely  such  that  the  signalling  dimension  of  the
composite  system is  larger  than the  product of  the
signalling  dimensions of  its  parts.   Our toy  model
theory  is  explicitly  derived   along  with  all  its
constituents: elementary and composite systems, states,
measurements, and dynamics.  In the process, we clarify
the  relation between  no-signaling, no-hypersignaling,
local tomography,  and information  causality, arriving
at the conclusion  that the no-hypersignaling principle
is independent of  all of these, and  must therefore be
assumed \textit{separately}.

The elementary system here is  the same as that used to
reproduce PR correlations  in Refs.~\cite{Bar05, Bar07,
  DT10}.   The states  and  effects  of the  elementary
system are  vectors in  $\mathbb{R}^3$, and  there only
exist four  extremal states and four  extremal effects,
namely   $\{   \omega_x   \}_{x=0}^3$   and   $\{   e_y
\}_{y=0}^3$.  As  shown explicitly in  the Supplemental
Material~\cite{SM} (see also Refs.~\cite{GMCD10, SB10})
all possible bipartite extensions can be given in terms
of $24$ extremal bipartite  states, namely $\{ \Omega_x
\}_{x = 0}^{23}$, and  $24$ extremal bipartite effects,
namely $\{  E_y \}_{y=0}^{23}$.  The first  $16$ states
(i.e.  $0 \le  x \le  15$) and  the first  $16$ effects
(i.e.  $0 \le  y  \le 15$)  are  factorized, while  the
remaining ones are all entangled.

Due  to  self-consistency   and  the  requirement  that
non-trivial   reversible   dynamics   exist,   however,
bipartite   states  and   effects   cannot  be   chosen
arbitrarily.  As  explicitly shown in  the Supplemental
Material~\cite{SM},  only  the following  three  models
satisfy all requirements:

\begin{description}[leftmargin=0pt]
\item[PR  Model]  this  is  the theory  used  to  model
  PR-boxes~\cite{Bar05, Bar07, DT10}.   It contains all
  possible  extremal  bipartite states,  including  the
  eight    entangled   ones    (i.e.    $\{    \Omega_x
  \}_{x=0}^{23}$).  Self-consistency  then imposes that
  only  extremal factorized  effects are  allowed (i.e.
  $\{ E_y \}_{y=0}^{15}$).
\item[HS Model] this is the  theory that we prove to be
  hypersignaling  (HS).   It contains  only  factorized
  extremal  states (i.e.   $\{\Omega_x \}_{x=0}^{15}$),
  but allows  for all  possible extremal  effects, even
  entangled ones (i.e. $\{E_y \}_{y=0}^{23}$).
\item[Hybrid  Models]  in  addition to  all  factorized
  states  and effects,  two  entangled  states and  two
  entangled  effects   are  allowed.   Self-consistency
  singles   out   only    two   such   models:   states
  $\{\Omega_{20},\Omega_{22}\}$       with      effects
  $\{E_{20},E_{22}\}$,             or            states
  $\{\Omega_{21},\Omega_{23}\}$ with effects $\{E_{21},
  E_{23}\}$.
\end{description}

Due to  the presence of bipartite  entangled states $\{
\Omega_x  \}_{x=16}^{23}$, the  PR model  is compatible
with super-quantum space-like correlations, and this is
actually the reason why it  was introduced in the first
place.     However,    we    show    in    Supplemental
Material~\cite{SM}, that the  lack of entangled effects
prevents the PR model  from being hypersignaling.  In a
perfectly  complementary  way,   the  HS  model  cannot
violate  any  Bell  inequality,  due  to  the  lack  of
entangled  states.  However,  in what  follows we  show
that,  due  to  the  presence  of  bipartite  entangled
effects $\{ E_y \}_{x=16}^{23}$,  the HS model violates
the no-hypersignaling principle.

Let  us   start  by  noticing  (see   the  Supplemental
Material~\cite{SM})  that the  elementary system  has a
signaling dimension  of two  and is thus  equivalent to
the  exchange  of  one classical  bit.   Therefore,  to
provide  a  counterexample   to  the  no-hypersignaling
principle, we need to provide a correlation $\xi$ which
is compatible with the composition of two elementary HS
systems, but cannot be  obtained by exchanging only two
classical bits.

One such a conditional probability has seven inputs and
seven outputs, and is given by
\begin{align}\label{eq:specimen}
  \xi = \frac{1}{2}
  \begin{pmatrix}
    1 & 0 & 0 & 0 & 0 & 1 & 0 \\
    0 & 1 & 0 & 0 & 0 & 0 & 1 \\
    0 & 1 & 1 & 0 & 0 & 0 & 0 \\
    0 & 0 & 1 & 0 & 0 & 0 & 1 \\
    0 & 0 & 0 & 1 & 0 & 1 & 0 \\
    0 & 0 & 0 & 1 & 0 & 0 & 1 \\
    0 & 0 & 0 & 0 & 1 & 1 & 0
  \end{pmatrix}.
\end{align}
This is  explicitly obtained by applying  the formalism
developed in the Supplemental Material~\cite{SM}.  More
explicitly,  the  rows  of $\xi$  are  the  conditional
probabilities  obtained  by   measuring  the  following
measurement: $\left  \{ \frac{1}{8}  E_{0}, \frac{1}{8}
  E_{1},   \frac{1}{8}    E_{6},   \frac{1}{8}   E_{8},
  \frac{1}{8}  E_{10}, \frac{1}{8}  E_{15}, \frac{1}{4}
  E_{23}  \right \},$  on each  of the  following seven
states: $ \left  \{ \Omega_{0}, \Omega_{2}, \Omega_{6},
  \Omega_{7},  \Omega_{12},   \Omega_{13},  \Omega_{15}
\right \}$.

The fact that $\xi$  does not belong to $\cc{7}{7}{4}$,
and  thus violates  the HS  principle, is  an immediate
consequence   of  the   characterization  of   polytope
$\cc{7}{7}{4}$    provided    in    the    Supplemental
Material~\cite{SM}.

Since  $\xi\not\in\cc{7}{7}{4}$,  there exists  a  game
which  violates  Eq.~\eqref{eq:iff}.  Indeed,  consider
the following game matrix $g$:
\begin{align*}
  g = \frac1{21}
  \begin{pmatrix}
    2 & 0 & 0 & 0 & 0 & 1 & 0 \\
    0 & 2 & 0 & 0 & 0 & 0 & 2 \\
    0 & 2 & 2 & 0 & 0 & 0 & 0 \\
    0 & 0 & 2 & 0 & 0 & 0 & 2 \\
    0 & 0 & 0 & 1 & 0 & 1 & 0 \\
    0 & 0 & 0 & 1 & 0 & 0 & 0 \\
    0 & 0 & 0 & 0 & 2 & 1 & 0
  \end{pmatrix}.
\end{align*}
It  immediately follows  by  explicit computation  that
$g^T  \cdot   \xi  =   \frac12$,  while   $\max_{q  \in
  \cc{7}{7}{4}} g^T  \cdot q  = \frac{10}{21}<\frac12$.
This  latter  result  can  be  verified  by  explicitly
computing the  payoff associated with game  $g$ for all
of the  vertices of the polytope  $\cc{7}{7}{4}$, which
are  $359863$ in  number as  shown in  the Supplemental
Material~\cite{SM}.  The interested reader can play the
game  of  selecting  $4$  columns of  $g$  and  further
selecting  one entry  per row  (within these  columns),
with  the aim  of maximizing  the sum  of the  selected
entries.  They  will then verify that  no strategy will
lead to a payoff larger than $\frac{10}{21}$.

\begin{figure}[t!]
\vspace{0.4cm}
  \begin{overpic}[width=.4\columnwidth]{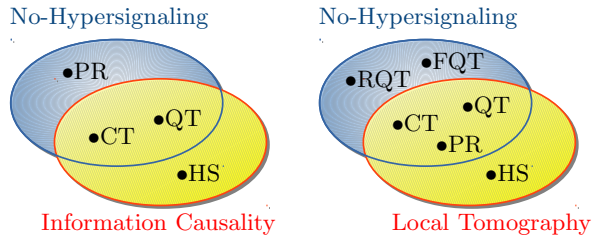}
    \put (0, 71) {\color{blue1}No-Hypersignaling}
    \put (12, -8) {\color{red}Information Causality}
    \put (20, 50) {$\bullet$PR}
    \put (30, 25) {$\bullet$CT}
    \put (55, 32) {$\bullet$QT}
    \put (64, 11) {$\bullet$HS}
  \end{overpic}\qquad
  \begin{overpic}[width=.4\columnwidth]{fig03}
    \put (0, 71) {\color{blue1}No-Hypersignaling}
    \put (28, -8) {\color{red}Local Tomography}
    \put (28, 30) {$\bullet$CT}
    \put (55, 37) {$\bullet$QT}
    \put (45, 22) {$\bullet$PR}
    \put (10,47) {$\bullet$RQT}
    \put (39, 54) {$\bullet$FQT}
    \put (64, 11) {$\bullet$HS}
  \end{overpic}
\vspace{0.2cm}
\caption{\textbf{No-Hypersignaling    vs    Information
    Causality  and  vs  Local Tomography}.   Left:  the
  diagram  compares   theories  satisfying  information
  causality  (yellow  set)  and  the  no-hypersignaling
  principle  (blue  set):  CT  (classical  theory),  QT
  (quantum theory), PR Model  (the toy model theory for
  PR-boxes),  and  HS  Model  (the  locally  classical,
  hypersignaling    theory    constructed    in    this
  paper).  Right: comparison  between local  tomography
  and  no-hypersignaling  as  two features  of  general
  probabilistic  theories.  Examples  of theories  that
  are non-hypersignaling  but violate  local tomography
  are  provided  by  real   quantum  theory  (RQT)  and
  fermionic  quantum theory  (FQT).   The  HS Model  is
  locally tomographic  but hypersignaling.  Finally CT,
  QT, and the PR Model lie in the intersection, as they
  obey both local  tomography and the no-hypersignaling
  principle.}
  \label{f:comparison}
\end{figure} 

\emph{Outlooks}.  ---  We have seen how  it is possible
to construct a generalized  probabilistic theory -- the
HS Model  -- that contradicts quantum  theory, but only
in time-like scenarios. This is consequence of the fact
that  the  HS Model  has  been  arranged so  that  only
separable  states  are  allowed.   In  this  way,  when
measurements  are restricted  to  be  separable due  to
locality constraints  (as it  is the case  when testing
space-like  correlations),  the  HS  Model  never  goes
beyond classical  theory.  However, the  possibility of
having  entangled measurements  enables hypersignaling,
thus   proving   that   the  HS   Model   indeed   goes
\textit{beyond} quantum theory in time-like scenarios.

It is now important to understand how hypersignaling is
logically  related with  other possible  ``anomalies,''
such  as  the  violation  of local  tomography  or  the
violation   of    information   causality.     If   any
hypersignaling theory  necessarily violates  also other
principles concerning space-like correlations, then one
could   rightly   argue    that   the   phenomenon   of
hypersignaling might  be ruled  out just by  looking at
space-like  correlations.  However,  the point  of this
paper  is   to  argue  the  opposite:   that  time-like
correlations   require   a   new   \textit{independent}
principle.

The fact that  hypersignaling and information causality
are  independent  is  easy   to  see.  As  a  necessary
condition for the violation of information causality is
the  presence of  entangled  states, and  since the  HS
Model only contains separable states, then the HS Model
necessarily   obeys   information  causality,   despite
allowing hypersignaling.  Vice versa,  we know that the
PR Model  violates information causality but,  since it
only allows  separable measurements, it  cannot display
any form of hypersignaling.   The situation is depicted
in Left Fig.~\ref{f:comparison}.

We    now   turn    to   the    condition   of    local
tomography~\cite{HW12}. From the explicit expression of
the  pure states  of the  HS Model,  it is  possible to
verify, as done in the Supplemental Material~\cite{SM},
that  the elementary  system $S$  has linear  dimension
$\ell(S)=3$ and that the  bipartite system $S\otimes S$
has linear  dimension $\ell(S\otimes  S )=9=\ell(S)^2$.
Thus the HS Model is locally tomographic, despite being
hypersignaling.   Vice  versa, there  exist  consistent
theories that obey  the no-hypersignaling principle and
yet are not locally tomographic.  As an example, let us
consider restrictions (for example, superselections) of
quantum  theory, as  introduced in  Ref.~\cite{DMPT14}.
Since such theories are restrictions of quantum theory,
they cannot  exhibit hypersignaling: if they  did, then
quantum theory would also exhibit hypersignaling, which
is    not   true.     For    example,   real    quantum
theory~\cite{HW12}      and      fermionic      quantum
theory~\cite{DMPT14} are  two possible  such restricted
quantum    theories.     However,    as    proved    in
Refs.~\cite{DMPT14,DMPT142,HW12}, both theories are not
locally  tomographic.  The  situation is  summarized in
Right Fig.~\ref{f:comparison}.

We also notice that the no-hypersignaling principle can
be   violated   by   theories    that   do   not   show
superadditivity    of    classical   capacities.     In
Ref.~\cite{MPP15}  the  authors  show  that  a  locally
tomographic   theory  cannot   feature  superadditivity
effects of  classical capacities.   Thus hypersignaling
does not necessarily imply superadditivity of classical
capacities,   because   the   HS   Model   is   locally
tomographic.   In   passing  by,  the   maximal  mutual
information for the hypersignaling correlation $\xi$ in
Eq.~(\ref{eq:specimen}) (numerically optimized over any
prior  probability distribution)  is  less than  $1.78$
bits, which is below the Holevo bound of $\log_24 = 2$.

One interesting question arises  from noting that while
the HS Model has  classical space-like correlations and
super-quantum time-like correlations,  the PR Model has
super-quantum  space-like  correlations  and  classical
time-like correlations.  Could it  be that a theory can
be super quantum only with respect to either space-like
or time-like correlations, but not both?  Could quantum
theory  have the  unique  distinction of  ``balancing''
between these two extrema? It turns out that the answer
is  no, and  follows  from the  example  of the  Hybrid
Models  derived   above.   In   order  to   obtain  the
hypersignaling $\xi$ in Eq.~(\ref{eq:specimen}) we need
seven factorized  states and seven effects  among which
only one, precisely $E_{23}$, is not factorized.  Since
$E_{23}$  is exactly  one  of  those entangled  effects
admitted in  the Hybrid Models,  we know that  the same
$\xi$  can  be surely  obtained  in  those models  too.
Moreover,  since in  the  Hybrid  Models two  entangled
states  are  also  available,  super-quantum  spacelike
correlations  can also  be created.  Hence, the  Hybrid
Models have  the ability to create  both space-like and
time-like super-quantum correlations.

Finally,  we  compare the  no-hypersignaling  principle
with two recently proposed and related principles, that
is,  {\em  dimension mismatch}~\cite{BKLS14}  and  {\em
  information   content}~\cite{CHHH14}.     Both   such
principles rule out superquantum  theories on the basis
of  the  correlations  achievable by  a  single-partite
system,   in   contrast  with   the   no-hypersignaling
principle which  requires composite  systems.  However,
they  achieve this  by considering  a more  complicated
setup,  where  the  choice  of the  information  to  be
decoded is not fixed but depends on an additional input
(a second  question) to  the receiver.   Moreover, both
the  dimension mismatch  principle and  the information
content  principle   rely  on   a  certain   degree  of
arbitrariness  in  the  criteria  chosen  to  benchmark
operational  theories:  dimension mismatch  is  defined
with respect  to an arbitrarily chosen  reference task,
i.e.  pairwise state  discrimination, while information
content  is  defined  with respect  to  an  arbitrarily
chosen     information     measure,    i.e.      mutual
information.   This    is   in   contrast    with   the
no-hypersignaling  principle proposed  here, where  the
full  set of  input-output  correlations is  considered
without   the    need   to   invoke    any   particular
discrimination task or information measure.

{\em  Acknowledgements}  The  authors are  grateful  to
Manik Banik for noticing two  oversights we made in the
definition  of  the two  hybrid  models  in a  previous
version  of   this  manuscript.    M.~D.   acknowledges
support  from  the   Singapore  Ministry  of  Education
Academic    Research   Fund    Tier   3    (Grant   No.
MOE2012-T3-1-009).  A.   T.  acknowledges  support from
the Templeton Foundation, No. 60609.  F.~B acknowledges
support  from the  Japan Society  for the  Promotion of
Science  (JSPS) KAKENHI,  Grant  No.  17K17796.   V.~V.
acknowledges support from the Ministry of Education and
the Ministry of Manpower (Singapore).



\section{Supplemental material}

Here we provide those technical results reported in the
letter ``No-hypersignaling  principle'' by  the present
authors  (M.  Dall'Arno,  S.  Brandsen,  A. Tosini,  F.
Buscemi, and  V. Vedral) that, not  being essential for
the presentation, were not included in the main text.

\subsection{General probabilistic theories}
\label{sec:gpt}

Generalized probabilistic theories  (GPTs) constitute a
very  general  framework,  suitable for  describing  an
arbitrary physical probabilistic  theory.  In this way,
the characteristic  quantum traits  can be  compared to
other  (in principle)  admissible behaviours,  with the
final goal  of seeking  for the physical  principles at
the basis of the quantumness  of nature.  As such, GPTs
proved to  be an  extremely useful  notion to  shed new
light on the apparently odd features of quantum theory.

The most  popular and  successful applications  of GPTs
aim   to   disclose   the  properties   of   space-like
correlations compatible with  special relativity and to
compare them  with the typical  space-like correlations
available   in   quantum   theory,  such   as   quantum
entanglement.    It   is   well  known   that   quantum
entanglement enables  two space-like  separated parties
to be correlated  in a way that would  be impossible if
only    classical    correlations    were    available.
Nevertheless,  quantum   correlations  are   still  non
signaling, in  the sense that they  cannot be exploited
to    give    instantaneous   (faster    than    light)
communication~\cite{EPR35}.   In  this respect,  hence,
quantum  correlations   are  compatible   with  special
relativity.

However, Ref.~\cite{PR94} shows that quantum space-like
correlations  are not  the  only  ones compatible  with
special  relativity,   but  that  there   exist  other,
super-quantum   and  yet   non  signaling,   space-like
correlations.  As noticed  by many authors \cite{Bar05,
  Bar07,  DT10}   such  super-quantum   non  signalling
correlations, usually referred  to as ``PR-boxes'', can
be interpreted in terms of a particular GPT.

The building  blocks of GPTs are  systems, here denoted
by  capital letters  $A, B,  C, \ldots$,  which can  be
composed  to   form  composite  systems,   for  example
$A\otimes B$ represents the composite system consisting
of subsystems  $A$ and  $B$. A system  $A$ is  given by
specifying how  it can  be prepared and  how it  can be
measured.  This  is  done  by giving  the  set  of  all
possible  states  and   all  possible  effects,  namely
$\mathcal{S}(A)$          and         $\mathcal{E}(A)$,
respectively. The  theory is then specified  by further
providing   a   complete  description   of   admissible
operations that any system in the theory can undergo.

Alongside  the mathematical  characterization of  these
sets, which ultimately defines  the theory, it is often
useful to  provide a graphical representation  for such
basic building blocks.  With systems depicted as wires,
each state is represented as
\begin{equation*}
  \begin{aligned}
    \Qcircuit @C=1em @R=.7em  @! R {\prepareC{\omega} &
      \poloFantasmaCn A\qw & \qw }
  \end{aligned} \; ,
\end{equation*}
where  $A$   is  the  system  prepared   in  the  state
$\omega\in   \mathcal{S}(A)$,   and  each   effect   is
represented as
\begin{equation*}
  \begin{aligned}
    \Qcircuit @C=1em  @R=.7em @! R {  & \poloFantasmaCn
      A\qw & \measureD{e} }
  \end{aligned} \; ,
\end{equation*}
where  $A$ is  the  system  undergoing the  observation
corresponding   to  the   effect  $e\in\mathcal{E}(A)$.
Composite  systems  are  then represented  by  multiple
parallel wires; for example
\begin{equation*}
  \begin{aligned}
    \Qcircuit @C=1em @R=.7em {
      \multiprepareC{1}{\Omega} &  \poloFantasmaCn A\qw & \qw \\
      \pureghost{\Omega} & \poloFantasmaCn B\qw & \qw }
  \end{aligned} \; ,
\end{equation*}
denotes the  bipartite state $\Omega$ of  the composite
system $A\otimes B$, while
\begin{equation*}
  \begin{aligned}
    \Qcircuit @C=1em @R=.7em @! R {
      & \poloFantasmaCn A \qw & \multimeasureD{1}{E}  \qw  \\
      & \poloFantasmaCn B \qw & \ghost{E} \qw }
  \end{aligned} \; ,
\end{equation*}
denotes  the  bipartite  effect $E$  of  the  composite
system $A \otimes B$.

The probabilistic  structure of  the theory  comes from
the rule that associates the probability $p_{e|\omega}$
of observing  any given effect $e  \in \mathcal{E}(A)$,
on      any      given      state      $\omega      \in
\mathcal{S}(A)$. Graphically, this is denoted as
\begin{equation*}
  \Qcircuit @C=1em @R=.7em @! R{ \prepareC{\omega}   &
    \poloFantasmaCn A
    \qw  & \measureD{e} 
  } \; = \; p_{e | \omega} \;.
\end{equation*}
Clearly,  any  closed   circuit,  however  complicated,
corresponds to a probability.

By construction,  states (resp., effects)  are positive
functionals on effects (resp., states):
\begin{align*}
  \omega:\mathcal{E}(A)\rightarrow         [0,1],\qquad
  e:\mathcal{S}(A)\rightarrow [0,1].
\end{align*} 
As such, it is  natural to consider linear combinations
of  states  and  linear  combinations  of  effects.  In
particular,  any convex  combination of  states (resp.,
effects) is itself an admissible state (resp., effect).
For this reason,  $\mathcal{S}(A)$ and $\mathcal{E}(A)$
are usually assumed to  be convex sets.  (The convexity
assumption can be relaxed  and theories with non-convex
state spaces,  such as Spekkens' toy  theory, have been
considered in the literature~\cite{S07}.)

Further, extending the linear combinations to arbitrary
real  coefficients,  one  can define  two  real  vector
spaces         $\mathcal{S}_{\mathbb{R}}(A)$        and
$\mathcal{E}_{\mathbb{R}}(A)$,  usually constructed  so
that one  is dual to  the other.  This means,  in other
words,  that   $\mathcal{E}_{\mathbb{R}}(A)$  coincides
with   the  set   of   all   linear  functionals   from
$\mathcal{S}_{\mathbb{R}}(A)$ to $\mathbb{R}$, and vice
versa.  (This  assumption is  sometimes referred  to as
the  no-restriction hypothesis,  and can  or cannot  be
made depending on the situation at hand.)

From these observations, a key feature of GPTs follows:
states and effects can always be represented as vectors
of a linear real space.   It is common then to restrict
to the  case of  GPTs whose set  of states  span finite
dimensional vector spaces.  In this case, and under the
no-restriction  assumption, one  can define  the linear
dimension   of   a   system    $A$   as   $\ell(A)   :=
\dim{\mathcal{S}_{\mathbb{R}}           (A)}          =
\dim{\mathcal{E}_{\mathbb{R}}(A)}$.

We can now introduce the notion of channels.  A channel
on  system $A$  (for  simplicity  we consider  channels
having the  same input and  output system) is  a linear
map $T$ from $\mathcal{S}(A)$ to itself
\begin{equation*}
  T: \omega\in \mathcal{S}(A)\mapsto T(\omega) \in \mathcal{S}(A).
\end{equation*}
and it is graphically represented as follows
\begin{equation*}
  \begin{aligned}
    \Qcircuit @C=1em @R=.7em @! R { \prepareC{\omega} &
      \qw   \poloFantasmaCn   A   &  \gate   {T}&   \qw
      \poloFantasmaCn A &\qw }
  \end{aligned}\:.
\end{equation*}
Moreover, for  any system  $C$, the map  $T\otimes I_C$
(with $I_C$ denoting the identity channel on the system
$C$,     namely    $I_C(\rho)     =\rho$    for     any
$\rho\in\mathcal{S}(C)$) must  correspond to  a channel
from $\mathcal{S}(A\otimes  C)$ to itself.   This means
that when  $T$ is  applied to a  subsystem of  a larger
bipartite  one, it  still  maps  bipartite states  into
bipartite  states.  The  last  condition resembles  the
condition of complete positivity of quantum channels.

In what  follows, we  will consider  a special  kind of
channels, namely \emph{reversible  channels}. These are
defined  as follows:  a channel  $U$ on  system $A$  is
reversible if and only  if there exists another channel
$U^{-1}$  such that  $UU^{-1}=U^{-1}U=I_A$.  We  denote
the set of all reversible channels as $\mathcal{U}(A)$.
Notice  that in  quantum theory  the set  of reversible
channels   coincides   with    the   set   of   unitary
transformations.

Another common assumption in the GPTs framework is that
of causality (the details  and the consequences of this
assumption on  the structure of a  probabilistic theory
can  be  found  in  Ref.~\cite{CDP11}).   A  theory  is
``causal'' if the choice of future measurement settings
does not  influence the outcome probability  of present
experiments.   Mathematically, the  causality condition
is equivalent to the fact that, for every system, there
exists only one deterministic  effect, denoted by $\bar
e$,   which  is   the  effect   that  has   conditional
probability equal to one on any state~\cite{CDP11}.

\subsection{Characterization of $\cc{m}{n}{d}$}
\label{sec:classcorr}

In this paper we  denote by $\cc{m}{n}{d}$ the polytope
of  all  $m$-input/$n$-output  conditional  probability
distributions $p_{y|x}$  that can be obtained  by means
of   one    $d$-dimensional   classical    or   quantum
system~\cite{FW15}.   Its  extremal  points  are  those
$p_{y|x}$ which are non-zero  for at most $d$ different
values of $y$, and their  non-zero entries are equal to
one.

Let     us     denote     with     $\binom{n}{k}     :=
\frac{n!}{k!(n-k)!}$ the binomial  coefficient and with
${m\brace         k}          :=         \sum_{j=0}^{k}
\frac{1}{k!}(-1)^{k-j}\binom{k}{j}  j^m$  the  Stirling
number  of  the  second   kind,  i.e.   the  number  of
partitions of  a set of  $m$ elements in  $k$ non-empty
classes.  Then the following result holds.

\begin{lmm}
  The number $V$ of vertices of $\cc{m}{n}{d}$ is equal
  to
  \begin{align*}
    V = \sum_{k=1}^{d} k! \binom{n}{k} {m\brace k}.
  \end{align*}
\end{lmm}

\begin{proof}
  The  statement   follows  by  a  a   simple  counting
  argument.  Let us arrange the numbers $p_{y|x}$ in an
  $m \times n$ stochastic  matrix, where $x$ labels the
  rows and  $y$ labels the columns.   One first chooses
  which are  the $k  \leq d$  non-null columns:  as the
  matrix has a total of $n$ possible columns, there are
  $\binom{n}{k}$ ways  to do so.  Then,  since each row
  consists only  of zeros and  a single one,  there are
  exactly $k!  {m \brace k}$ possible arrangements. The
  factorial  $k!$ comes  from  the fact  that here  the
  order  of  the  partition   is  relevant,  while  the
  definition of the Stirling  number of the second kind
  does not take this into account.
\end{proof}

\subsection{Witnessing      violations      of      the
  no-hypersignaling principle}
\label{sec:violations}

A measurement $M$ is a  family of effects summing up to
the unit effect, namely the effect that has conditional
probability one  given any  state. A convenient  way to
represent  any   measurement  is  the   following.   By
rescaling each  effect by  a positive  coefficient, one
can have  all the effects of  the theory to lie  in the
hyperplane  that  contains  the   unit  effect  and  is
orthogonal to the unit  effect.  Although such rescaled
effects can  be out of  the truncated cone  of effects,
they  are  all  linear  combinations  of  effects  with
positive   coefficients  (hence,   they   lie  in   the
non-truncated cone of effects).

For the aforementioned reason,  and with a slight abuse
of  notation, we  will  refer to  them as  ``normalized
effects''.   In  other  words, each  normalized  effect
identifies the class of equivalence of all effects that
lie  on  the same  ray,  obtained  by intersecting  the
truncated   cone   with   the   hyperplane.    Extremal
normalized effects are normalized effects that also lie
on  an  extremal  ray  of the  cone  of  effects.   For
example,   in  quantum   theory  this   corresponds  to
rescaling any effect so that  they all have trace equal
to the dimension  of the Hilbert space  (the same trace
of the unit effect).

This allows one to represent any measurement $M$ simply
as  a probability  distribution  $p_y$ over  normalized
effects $\{e_y\}$, with the  condition that $\sum_y p_y
e_y  = \bar{e}$,  where $\bar{e}$  is the  unit effect.
The effects of such a measurement are just given by the
normalized  effects   weighted  by   the  corresponding
probabilities, that is $p_y e_y$.  We then say that the
family     of      normalized     effects     $\{e_y\}$
\textit{supports} the measurement $M$.  For example, in
quantum theory, a measurement is a POVM.

The above representation turns out to be very useful in
our analysis  for the following reason.   By linearity,
when  looking for  violations of  the no-hypersignaling
principle as  given by Eq.~\eqref{eq:iff},  it suffices
to consider  extremal families  of states  and extremal
measurements.     While   the    former   are    simply
characterized  as  families  of  extremal  states,  the
characterization  of  extremal   measurements  is  more
complicated.     Generally,    there    are    extremal
measurements  whose supporting  normalized effects  are
not  all extremal  (this  is also  a  known feature  of
quantum theory~\cite{Par99, DLP05}).

However, as shown below, when looking for violations of
the   no-hypersignaling  principle,   it  suffices   to
consider  extremal measurements  supported by  extremal
normalized effects.  This fact, together with the above
representation,  allows  us  to write  any  measurement
potentially  violating the  no-hypersignaling principle
simply  as  a  probability distribution  over  extremal
normalized effects, which are finite in number and thus
easily  characterizable.   This,  in turn,  provides  a
efficient  way  to check  whether  a  GPT containing  a
finite  number   of  extremal  normalized   effects  is
hypersignaling or not.

We start by  showing that, for the problem  at hand, it
suffices   to  consider   extremal  measurements   with
extremal normalized effects:
\begin{thm}
  \label{thm:exteffects}
  If a composite system $S = \otimes_kS_k$ violates the
  no-hypersignaling principle, then  a violation occurs
  for   some  measurement   with  extremal   normalized
  effects.
\end{thm}

\begin{proof}
  Suppose  that  a  certain  payoff $g^T  \cdot  p$  is
  obtained  by means  of  a measurement  $M$ whose  $n$
  normalized  effects  are  not  all  extremal.   Since
  normalized effects  of $S$  are $\ell(S)$-dimensional
  real vectors with  a linear normalization constraint,
  by  Caratheodory's  theorem,  each  of  them  can  be
  decomposed  as  the  convex combination  of  at  most
  $\ell(S)$ extremal normalized effects.
  
  Consider hence  a refinement of measurement  $M$ into
  another measurement  $M'$ with $n' :=  \ell(S) \times
  n$ extremal normalized  effects.  Correspondingly, we
  also expand game $g$ into another game $g'$, which is
  obtained  from $g$  by writing  $\ell(S)$ times  each
  column. Trivially, by construction, the payoff of $M$
  for game  $g$ is the same  as the payoff of  $M'$ for
  game $g'$.

  Let us now show that games $g$ and $g'$ have the same
  classical payoff, namely the same ``no-hypersignaling
  threshold'', i.e.
  \begin{align*}
    \max_{p \in  \cc{m}{n}{d}} g^T  \cdot p  = \max_{p'
      \in \cc{m}{n'}{d}} g'^T \cdot p' \; .
  \end{align*}
  This is proved since the right hand side is obviously
  not smaller than the left hand side (as game $g$ is a
  ``sub-game''  of game  $g'$).   Conversely, the  left
  hand side is not smaller than the right hand side, as
  a consequence  of the  fact that  any coarse-graining
  identifying   all   the  outcomes   associated   with
  identical columns of $g'$ transforms any $n'$-outcome
  measurement attaining  some payoff  for $g'$  into an
  $n$-outcome  measurement  attaining  the  {\em  same}
  payoff for $g$.

  The above  arguments hence show that,  if a violation
  of the no-hypersignaling threshold is observed with a
  measurement with non-extremal normalized effects, the
  same   violation  can   be  observed   also  with   a
  measurement  supported  only by  extremal  normalized
  effects.
\end{proof}

Let us now provide  a full closed-form characterization
of  the  set  of extremal  measurements  with  extremal
normalized   effects    in   any    given   generalized
probabilistic theory:
\begin{thm}
  \label{thm:lineffects}
  Any  measurement  $M =  \{  p_y  >  0, e_y  \}$  with
  extremal normalized  effects $\{ e_y \}$  is extremal
  if and only if $\{ e_y \}$ are linearly independent.
\end{thm}

\begin{proof}
  Let us  first prove the  ``only if'' part. By  way of
  contradiction, let us assume there exists an extremal
  measurement $\{ p_y,  e_y \}$ with $p >  0$ such that
  $\{  e_y \}$  are not  linearly independent,  i.e. $|
  \supp(p) | > \dim\spn(\{ e_y \})$.  Since $\{ e_y \}$
  are normalized, they belong  to an affine subspace of
  dimension $\dim\spn(\{ e_y \})  - 1$.  Thus, applying
  Caratheodory's  theorem, the  unit effect  $\bar{e}$,
  which we know to belong  to $\spn(\{ e_y \})$, can be
  decomposed in  terms of a  subset of the $\{  e_y \}$
  with cardinality  $\dim\spn(\{ e_y \})$,  i.e.  there
  exists a  probability $p'_y$ with $|  \supp(p') | \le
  \dim\spn(\{  e_y \})$  such that  $\sum_y p'_y  e_y =
  \bar{e}$.  By  taking $\lambda  > 0$  such that  $p -
  \lambda  p' \ge  0$ (such  a $\lambda$  always exists
  since $p > 0$) and  $p''_y := (1-\lambda)^{-1} (p_y -
  \lambda p'_y)$, it immediately  follows that also $\{
  p''_y, e_y \}$  is a measurement.  Then  $\{ p_y, e_y
  \}$ can  be decomposed as  $\lambda \{p'_y, e_y  \} +
  (1-\lambda)  \{p''_y,  e_y  \}$,   i.e.   it  is  not
  extremal, thus leading to a contradiction.

  Let us now prove the  ``if'' part.  Since $\{ e_y \}$
  are extremal,  they cannot be further  decomposed, so
  any  convex decomposition  of  $M$ would  necessarily
  involve subsets of $\{ e_y \}$. Since $\{ e_y \}$ are
  linearly independent, the  decomposition of $\bar{e}$
  is unique,  and since $p >  0$ any subset of  $\{ e_y
  \}$ cannot be a measurement.  Therefore the statement
  follows.
\end{proof}

As       an        immediate       consequence       of
Theorems~\ref{thm:exteffects} and~\ref{thm:lineffects},
one has the following:
\begin{cor}
  \label{cor:counterexample}
  If a composite system $S = \otimes_kS_k$ violates the
  no-hypersignaling principle, then  a violation occurs
  for  some  measurement  with $n$  extremal,  linearly
  independent normalized effects, with
  \begin{align*}
    \prod_k\kappa(S_k) < n \le \ell(\otimes_k S_k).
  \end{align*}
\end{cor}

\begin{proof}
  Sufficiency of  measurements with  extremal, linearly
  independent  normalized  effects immediately  follows
  from                    Theorems~\ref{thm:exteffects}
  and~\ref{thm:lineffects}.    The   first   inequality
  immediately   follows   from   the  fact   that   any
  $m$-inputs/$n$-outputs  correlation $p$  with $n  \le
  \prod_k\kappa(S_k) =: K$ belongs to $\cc{m}{n}{K}$ by
  definition. The second inequality immediately follows
  from   the  condition   of  linear   independence  of
  normalized effects.
\end{proof}

For  any composite  system  $S =  \otimes_kS_k$ with  a
finite   number   of   extremal   normalized   effects,
Corollary~\ref{cor:counterexample}      provides     an
efficient    way   to    find    violations   of    the
no-hypersignaling principle.   For any set $\{  E_y \}$
with     $n$    normalized     effects    such     that
$\prod_k\kappa(S_k) < n  \le \ell(\otimes_kS_k)$ (these
sets are finite in number), one proceeds as follows:

\begin{enumerate}
\item Check if  $\{ E_y \}$ supports  a measurement.  A
  set $\{ E_y \}$ supports a measurement if and only if
  the following linear program is feasible
  \begin{align*}
    \min_{\substack{p \\ \sum_y p_y E_y = \bar{E} \\ p > 0}} 0,
  \end{align*}
  (notice that the objective function is irrelevant, as
  one is only interested in feasibility).
\item Check  if $\{  E_y \}$ are  linearly independent.
  Normalized   effects  $\{   E_y   \}$  are   linearly
  independent if  and only if  the matrix with  $\{ E_y
  \}$ as columns is full rank.
\end{enumerate}
Then one  has that  $S$ violates  the no-hypersignaling
principle if and only if  a violation occurs for one of
the sets of normalized effects  that passed both of the
two above checks.

Finally,  notice that  for  a  theory satisfying  local
tomography,  the  statement  of the  Corollary  further
simplifies:  the  number   $n$  of  extremal,  linearly
independent normalized effects is bounded as follows
\begin{align*}
  \prod_k\kappa(S_k) < n \le \prod_k\ell(S_k).
\end{align*}

\subsection{Construction of a class of toy models}
\label{sec:models}

Here we restrict to the  simple case of theories with a
single  ``type'' of  elementary system  $S$. We  assume
that the  system $S$ has linear  dimension $\ell(S)=3$,
namely  its   states  $\omega$  and  effects   $e$  are
described     by      vectors     in     $\mathbb{R}^3$
($\mathcal{S}=\mathcal{E}=\mathbb{R}^3$).     We    now
specify    the    convex   sets    $\mathcal{S}$    and
$\mathcal{E}$.   The  system  $S$ has  only  four  pure
(extremal) states,
\begin{align*}
  \omega_0 =\begin{pmatrix} 1\\0\\1
\end{pmatrix},\:
  \omega_1 =\begin{pmatrix}
0\\1\\1
\end{pmatrix},\:
  \omega_2 =\begin{pmatrix}
-1\\0\\1
\end{pmatrix},\:
  \omega_3 =\begin{pmatrix}
0\\-1\\1
\end{pmatrix}.
\end{align*}
The convex  set of states is  geometrically represented
by  a square  (see the  square  in the  plane $z=1$  in
Fig.~\ref{fig:elementary-system}) whose finite group of
symmetries  (the dihedral  group of  order eight  $D_8$
containing   four  rotations   and  four   reflections)
coincides with  the set of reversible  channels for the
system $S$, explicitly given by
\begin{equation}\label{eq:single-system-unitaries}
\begin{aligned}
  \mathcal{U}(S)=\{U_k^s: k=0,\ldots,3, s=\pm\}\\
  U_k^s =
  \begin{pmatrix}
    \cos \frac{\pi k}2 & -s \sin \frac{\pi k}2 & 0 \\
    \sin \frac{\pi k}2 & s \cos \frac{\pi k}2 & 0 \\
    0 & 0 & 1
  \end{pmatrix}.
\end{aligned}
\end{equation}
The matrices $U_k^{+}$ and $U_k^{-}$ represent the four
rotations and the four reflections respectively.

Assuming that  the probability associated to  an effect
on a state  is given by the trace  rule, we immediately
characterize  the  set  of  extremal  effects  for  the
elementary systems, namely the  set of vectors $e$ such
that $\Tr[e^T  \omega] \ge  0$ for any  state $\omega$.
This  leads  to  the   truncated  cone  of  effects  in
Fig.~\ref{fig:elementary-system},     with     extremal
normalized effects given by
\begin{align*}
  e_0 =\begin{pmatrix}
    1\\1\\1
  \end{pmatrix},\:
  e_1 =\begin{pmatrix}
    -1\\1\\1
  \end{pmatrix},\:
  e_2 =\begin{pmatrix}
    -1\\-1\\1
  \end{pmatrix},\:
  e_3 =\begin{pmatrix}
    1\\-1\\1
  \end{pmatrix}.
\end{align*}
Notice that the condition  $\Tr[e^T \omega] \leq 1$ for
any state $\omega$ and effect $e$ implies that extremal
effects are  obtained by  dividing by $2$  the extremal
normalized effects.  It is  immediate to check that the
deterministic effect, namely  the effect $\bar{e}$ such
that   $\Tr[\bar   e^T   \omega]=1$   for   any   state
$\omega\in\mathcal{S}$, must  be the vector  $\bar{e} =
(0, 0, 1)^T$.

Notice that extremal states (resp., normalized effects)
can be written in terms  of an arbitrary extremal point
via the reversible channels of the elementary system in
Eq.~\eqref{eq:single-system-unitaries}, for  example we
can write $\omega_x =  U_x^+ \omega_0^T$ (resp., $e_y =
U_y^+ e_0^T$).

We now  consider the  bipartite system $S\otimes  S$ of
linear  dimension $\ell(S\otimes  S)=9$  and thus  with
states $\Omega$ and  normalized effects $E$ represented
by  vectors  in $\mathbb{R}^9$.   The  goal  is now  to
derive  the self-consistent  bipartite GPTs  compatible
with the above elementary system $S$.

\begin{figure}[htb]
  \begin{overpic}[width=.7\columnwidth]{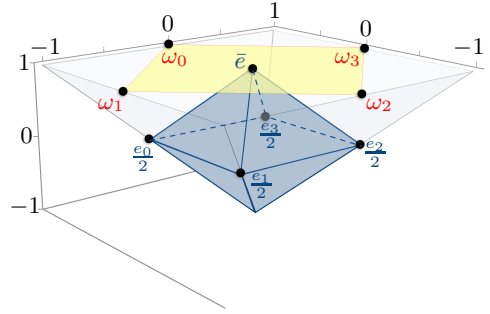}
    \put (15, 44) {\color{red}$\omega_1$}
    \put (67, 54) {\color{red}$\omega_3$}
    \put (29, 54) {\color{red}$\omega_0$}
    \put (74, 44) {\color{red}$\omega_2$}
    \put (48, 26) {\color{blue1}$\frac{e_1}2$}
    \put (73, 33) {\color{blue1}$\frac{e_2}2$}
    \put (22, 32) {\color{blue1}$\frac{e_0}2$}
    \put (50, 37) {\color{blue1}$\frac{e_3}2$}
    \put (45, 53) {\color{blue1}$\bar e$}
    \put (72.5, 60) {\color{black}$0$}
    \put (52.5, 64) {\color{black}$1$}
    \put (93, 56) {\color{black}$-1$}
    \put (-2.5, 52.5) {\color{black}$1$}
    \put (-2, 36.5) {\color{black}$0$}
    \put (-4.5, 21) {\color{black}$-1$}
    \put (0, 56) {\color{black}$-1$}
    \put (29, 61) {\color{black}$0$}
\end{overpic}
\caption{{\bf  Elementary system  for  a  class of  toy
    models.}   This  picture   depicts  the   ``squit''
  elementary  system  often considered  in  generalized
  probabilistic theories (in analogy to the ``bit'' and
  the  ``qubit''   which  are  elementary   systems  of
  classical  and  quantum  theory  respectively).   The
  system  is  fully specified  by  its  sets of  states
  (preparations)   and   effects  (observations)   here
  represented as vectors in $\mathbb{R}^3$.  The convex
  set of normalized states is represented by the yellow
  square at  the top, while  the convex set  of effects
  corresponds to the truncated blue cone.}
  \label{fig:elementary-system}
\end{figure}

It is  a convenient standard practice  to represent the
states  and  effects  of   two  elementary  systems  as
$3\times  3$ real  matrices rather  than as  vectors in
$\mathbb{R}^9$.   Any   bipartite  extension  naturally
includes  the  $16$   factorized  extremal  states  and
normalized effects given by
\begin{align*}
  \Omega_{4i+j} := \omega_i \otimes \omega_j^T,\qquad
  E_{4i+j} := e_i \otimes e_j^T,
\end{align*}
where $i, j \in \{ 0, 1, 2, 3 \}$.

Moreover, one  can introduce  other matrices  that play
the role  of entangled states and  effects.  These must
be compatible  with all  factorized effects  and states
given above. Diagrammatically, any such candidate state
$\Omega$  and normalized  effect $E$  must satisfy  the
following:
\begin{equation*}
\begin{aligned}
  &\begin{aligned} \Qcircuit @C=1em @R=.7em @! R {
      \multiprepareC{1}{\Omega} & \qw\poloFantasmaCn S   & \measureD{e_j} \\
      \pureghost{\Omega}   &  \qw\poloFantasmaCn   S  &
      \measureD{e_{j'}} }
  \end{aligned}\geq 0,\qquad & \forall  j, j' \in \{ 0,
  1, 2, 3 \},
  \\
  \\
  &\begin{aligned} \Qcircuit @C=1em @R=.7em @! R {
      \prepareC{\omega_i} & \poloFantasmaCn S \qw  & \multimeasureD{1}{E} \\
      \prepareC{\omega_{i'}}  & \poloFantasmaCn  S\qw &
      \ghost{E} }
  \end{aligned}\geq 0,\qquad & \forall  i, i' \in \{ 0,
  1, 2, 3\}.
\end{aligned}
\end{equation*}

It  is lengthy  but not  difficult to  verify that  the
matrices
\begin{align*}
  \Omega_{16} & := \frac12
  \begin{pmatrix}
    -1 & 1 & 0 \\
    1 & 1 & 0 \\
    0 & 0 & 2
  \end{pmatrix},
  & \Omega_{17} & := \frac12
  \begin{pmatrix}
    -1 & -1 & 0 \\
    -1 & 1 & 0 \\
    0 & 0 & 2
  \end{pmatrix}, \\
  \Omega_{18} & := \frac12
  \begin{pmatrix}
    1 & -1 & 0 \\
    -1 & -1 & 0 \\
    0 & 0 & 2
  \end{pmatrix},
  & \Omega_{19} &:= \frac12
  \begin{pmatrix}
    1 & 1 & 0 \\
    1 & -1 & 0 \\
    0 & 0 & 2
  \end{pmatrix}, \\
  \Omega_{20} & := \frac12
  \begin{pmatrix}
    -1 & -1 & 0 \\
    1 & -1 & 0 \\
    0 & 0 & 2
  \end{pmatrix},
  & \Omega_{21} & := \frac12
  \begin{pmatrix}
    1 & -1 & 0 \\
    1 & 1 & 0 \\
    0 & 0 & 2
  \end{pmatrix}, \\
  \Omega_{22} & := \frac12
  \begin{pmatrix}
    1 & 1 & 0 \\
    -1 & 1 & 0 \\
    0 & 0 & 2
  \end{pmatrix},
  & \Omega_{23} & := \frac12
  \begin{pmatrix}
    -1 & 1 & 0 \\
    -1 & -1 & 0 \\
    0 & 0 & 2
  \end{pmatrix},
\end{align*}
and 
\begin{align*}
  E_{16} & :=
  \begin{pmatrix}
    -1 & 1 & 0 \\
    1 & 1 & 0 \\
    0 & 0 & 1
  \end{pmatrix}, &
  E_{17} & :=
  \begin{pmatrix}
    -1 & -1 & 0 \\
    -1 & 1 & 0 \\
    0 & 0 & 1
  \end{pmatrix}, \\
  E_{18} & :=
  \begin{pmatrix}
    1 & -1 & 0 \\
    -1 & -1 & 0 \\
    0 & 0 & 1
  \end{pmatrix}, &
  E_{19} & :=
  \begin{pmatrix}
    1 & 1 & 0 \\
    1 & -1 & 0 \\
    0 & 0 & 1
  \end{pmatrix}, \\
  E_{20} & :=
  \begin{pmatrix}
    -1 & 1 & 0 \\
    -1 & -1 & 0 \\
    0 & 0 & 1
  \end{pmatrix}, &
  E_{21} & :=
  \begin{pmatrix}
    1 & 1 & 0 \\
    -1 & 1 & 0 \\
    0 & 0 & 1
  \end{pmatrix}, \\
  E_{22} & :=
  \begin{pmatrix}
    1 & -1 & 0 \\
    1 & 1 & 0 \\
    0 & 0 & 1
  \end{pmatrix}, &
  E_{23} & :=
  \begin{pmatrix}
    -1 & -1 & 0 \\
    1 & -1 & 0 \\
    0 & 0 & 1
  \end{pmatrix},
\end{align*}
satisfy  the  above  requirements, i.e.   they  satisfy
$\Tr[E_j^T \Omega_i] \ge 0$ for any $i \in [0, 15]$ and
$j \in [16,  23]$, and for any $i \in  [16, 23]$ and $j
\in [0, 15]$.  It is  also lengthy but not difficult to
verify  that  no  other  bipartite  extremal  state  or
normalized  effect  is  allowed.   It  is  possible  to
express in terms  of the reversible channels given
in Eq.\eqref{eq:single-system-unitaries}, namely
\begin{align*}
  \Omega_{16+k} & = \Omega_{16} {U_k^+}^T, & \Omega_{20+k}
  & = \Omega_{16} {U_k^-}^T, \\
  E_{16+k} & = U_k^+ E_{16}, & E_{20+k} & = U_k^- E_{16}.
\end{align*}
Finally,  the deterministic  effect  for the  bipartite
system is $\bar E=\bar e\otimes \bar e^T$.

In general the consistency of the theory (positivity of
the predicted  probabilities), imposes  restrictions on
the  admissible  entangled  states and  effects.   With
this, we mean that  any well-formed closed circuit must
give  rise   to  non-negative  probabilities.   Let  us
consider now this particular circuit:
\begin{align}
  \label{eq:smartcircuit}
  \begin{aligned}
    \Qcircuit     @C=1em     @R=.7em     @!     R     {
      \multiprepareC{1}{\Omega_x}  & \poloFantasmaCn  S
      \qw & \gate{U^{n_0}} & \qw&\poloFantasmaCn S\qw &
      \multimeasureD{3}{E_y} \\
      \pureghost{\Omega_x}  &  \poloFantasmaCn  S\qw  &
      \gate{U^{n_1}} &\poloFantasmaCn S\qw &
      \multimeasureD{1}{E_y} & \pureghost{E_y} \\
      \multiprepareC{1}{\Omega_x}   &   \poloFantasmaCn
      S\qw &
      \gate{U^{n_2}} &\poloFantasmaCn S\qw & \ghost{E_y} &      \pureghost{E_y} \\
      \pureghost{\Omega_x}  & \poloFantasmaCn  S \qw  &
      \gate{U^{n_3}}  &  \qw  &\poloFantasmaCn  S\qw  &
      \ghost{E_y} }
  \end{aligned} \; ,
\end{align}
where $U := \Omega_x E_y\in\mathcal{U}(S)$, namely
\begin{align*}
  U =
  \begin{aligned}
    \Qcircuit @C=1em @R=.7em @! R {
      \multiprepareC{1}{\Omega_x} & \qw  \\
      \pureghost{\Omega_x} & \multimeasureD{1}{E_y} \\
      & \ghost{E_y} }
  \end{aligned} \;,
\end{align*}
and  $n_i =  0,  1$.  The  requirement  that the  above
circuit  generates  non-negative   probabilities  is  a
necessary condition for any  choice of bipartite states
and bipartite effects to be consistent.  In formula, we
need to check the following inequality:
\begin{align}
  \label{eq:consistency}
  \Tr[ {U^{n_3}}^T  E_y^T U^{n_0}  \Omega_x {U^{n_1}}^T
  E_y U^{n_2} \Omega_x] \ge 0.
\end{align}
By explicit computation for $x \in [16, 23]$ and $y \in
[16,  23]$, the  only pairs  $(x, y)$  that generate  a
non-negative probability  in Eq.~\eqref{eq:consistency}
are the eight  pairs where $x = y$,  and the additional
four combinations  $(20, 22)$, $(22, 20)$,  $(21, 23)$,
and $(23, 21)$.

The  circuit  in Eq.~\eqref{eq:smartcircuit}  alone  is
enough to  leave us with  only four models  (apart from
trivial submodels), having  the following prescriptions
for the sets of  pure states $\{\Omega_i\}$ and effects
$\{E_j\}$:
\begin{enumerate}
\item {\bf  PR model:} All the  24 states $i\in[0,23]$;
  only the 16 factorized effects $j \in [0, 15]$;
\item {\bf HS  model:} Only the 16  factorized states $
  i\in[0,15]$; all the 24 effects $j \in [0, 23]$;
\item {\bf Hybrid models:}  Only 2 entangled states and
  effects  are included,  i.e.  $i  \in[0,15]\cup \{20,
  22\}$   and  $j\in[0,15]\cup   \{20,   22\}$  or   $i
  \in[0,15]\cup \{21,  23\}$ and  $j\in[0,15]\cup \{21,
  23\}$;
\item {\bf Frozen Models:} Only one entangled state and
  effect is  included, i.e.  $i \in[0,15]  \cup \{i'\}$
  and   $j\in[0,15]\cup  \{j'\}$   with  $i'   =  j'\in
  [16,23]$.
\end{enumerate}
One  can now  easily  verify that  within the  selected
models, any other circuit gives positive probabilities.
We  are  now  in  a   position  to  fully  specify  the
reversible dynamics  $\mathcal{U}(S\otimes S)$  for the
bipartite  system  $S\otimes  S$, which  follows  as  a
simple    consequence   of    the   main    result   of
Ref.~\cite{GMCD10}.   One   has  that   any  reversible
channel  corresponds to  the tensor  product of  single
system   reversible   channels,   possibly   with   the
application of the \emph{swap}  map $W$, namely the map
that exchanges the two subsystems. In formula one has
\begin{equation}\label{eq:reversible}
  \begin{aligned}
    \mathcal{U}(S\otimes S)\subseteq  &\{ W^i(U_j^{s_1}
    \otimes
    U_k^{s_2})\}\\
    &i=0,1,\; 0\leq j,k\leq 3,\; s_1,s_2=\pm.
  \end{aligned}
\end{equation}
Therefore,    reversible    channels   cannot    create
entanglement,  i.e.   transform factorized  states  and
effects   into    factorized   states    and   effects,
respectively.   This  creates a  clear-cut  distinction
between factorized and  entangled states (and effects),
as the ones cannot be mapped into the others.

Now   we   can    exploit   the   characterization   in
Eq.~\eqref{eq:reversible}   to  specify   the  set   of
reversible channels,  defined as the largest  subset of
$\mathcal{U}(S  \otimes   S)$  that  keeps   the  model
self-consistent.  By direct inspection we get
\begin{enumerate}
\item   PR  Model:   $\mathcal{U}_{PR}(S\otimes  S)   =
  \mathcal{U}(S \otimes S)$;
\item   HS  Model:   $\mathcal{U}_{HS}(S\otimes  S)   =
  \mathcal{U}(S \otimes S)$;
\item  Hybrid Models:  $\mathcal{U}_{HY}(S\otimes S)  =
  \{U_{k_1}^+      \otimes     U_{k_2}^+\}$,      where
  $k_1,k_2 = 0, 2$;
\item  Frozen Models:  $\mathcal{U}_{FR}(S\otimes S)  =
  \{W^i(U_0^+  \otimes  U_0^+)\}$,   where  $i=0,1$  if
  $x\in\{16,17,18,19\}$, while $i=0$ otherwise.
\end{enumerate}

We can  now justify  the name ``Frozen  Models'', since
these  models  comprise  only  the  trivial  reversible
dynamics.

The focus  of this manuscript  is on the HS  Model that
can be regarded  as the counterpart of the  PR Model in
the following  sense.  In Ref.~\cite{SB10}  the authors
point out  the existence of a  trade-off between states
and  effects in  the  PR Model  and  more generally  in
arbitrary non-local  theories colloquially  referred to
as  box world.  They show  that while  box world  allow
states whose space-like  correlations are stronger than
quantum theory, measurements in  box world are limited:
in  the  PR  Model   only  factorized  effects  can  be
observed.   On the  contrary,  measurements  in the  HS
Model can  contain entangled  effects, at the  price of
excluding all entangled states.

\subsection{Extremal measurements of the HS Model}
\label{sec:extmeas}

According        to        the        results        of
Section~\ref{sec:violations}, the HS Model violates the
no-hypersignaling principle if and  only if a violation
occurs  for  an   extremal  measurement  with  extremal
normalized effects.   Thus, we now turn  to the problem
of characterizing such measurements.

As  consequence of  Corollary~\ref{cor:counterexample},
for  the  elementary  system  $S$ there  are  only  two
possible extremal measurements with extremal normalized
effects, namely,  $\{ e_0,  e_1\}$ and $\{  e_2, e_3\}$
with uniform distribution $p = \frac12$.  Therefore, by
definition,  $S$ is  equivalent  to the  exchange of  a
classical bit.

Again         as        a         consequence        of
Corollary~\ref{cor:counterexample},  for the  bipartite
system   $S  \otimes   S$   there   are  fifteen   such
measurements   (modulo  equivalence   under  reversible
transformations).       They     are      listed     in
Table~\ref{tab:extmeas}  and  labeled   from  $M_0$  to
$M_{14}$.  The number of effects in each measurement is
indicated by the symbol $\#$  in the second column.  In
formula, the set of extremal measurements with extremal
normalized effects is given by
\begin{equation*}
  \label{eq:extmeas}
  \mathcal{M}_{HS}:=\{M_n\}_{n=0}^{14},\qquad
  M_n:=\{ p_y ,U E_y \},
\end{equation*}
where $U$ is any element of $\mathcal{U}_{HS} (S\otimes
S)  = \mathcal{U}(S\otimes  S)$  and the  probabilities
$p_y$ are explicitly listed in Table~\ref{tab:extmeas}.

\begin{table*}[htb!]
  \centering
  \begin{tabular}{| >{$}c<{$} | >{$}c<{$} ? >{$}c<{$} | >{$}c<{$} | >{$}c<{$} | >{$}c<{$} | >{$}c<{$} | >{$}c<{$} | >{$}c<{$} | >{$}c<{$} | >{$}c<{$} | >{$}c<{$} | >{$}c<{$} | >{$}c<{$} | >{$}c<{$} | >{$}c<{$} | >{$}c<{$} | >{$}c<{$} ? >{$}c<{$} | >{$}c<{$} | >{$}c<{$} | >{$}c<{$} | >{$}c<{$} | >{$}c<{$} | >{$}c<{$} |  >{$}c<{$} |}
    \hline
    \bf M &\bf \# & \bf E_0 & \bf E_1 & \bf E_2 & \bf E_3 & \bf E_4 & \bf E_5 & \bf E_6 & \bf E_7 & \bf E_8 & \bf E_9 & \bf E_{10} & \bf E_{11} & \bf E_{12} & \bf E_{13} & \bf E_{14} & \bf E_{15} & \bf E_{16} & \bf E_{17} & \bf E_{18} & \bf E_{19} & \bf E_{20} & \bf E_{21} & \bf E_{22} & \bf E_{23}\\
    \hline\hline
    0 & 2 &&&&&&&&&&&&&&&&& \frac12 && \frac12 &&&&& \\
    \hline
    1 & 4 &\nicefrac14&&\nicefrac14&&&&&&\nicefrac14&&\nicefrac14&&&&&&&&&&&&& \\
    \hline
    2 & 4 &\nicefrac14&&\nicefrac14&&&&&&&\nicefrac14&&\nicefrac14&&&&&&&&&&&& \\
    \hline
    3 & 6 &\nicefrac18&\nicefrac18&&&&&&&&&\nicefrac18&\nicefrac18&&&&&&&\nicefrac14&&&&&\nicefrac14 \\
    \hline
    4 & 6 &\nicefrac18&&&&&\nicefrac18&&&&&\nicefrac18&&&&&\nicefrac18&&&&&\nicefrac14&&&\nicefrac14 \\
    \hline
    5 & 6 &\nicefrac16&&&&&&&&&&\nicefrac16&&&&&&&\nicefrac16&\nicefrac16&&\nicefrac16&&&\nicefrac16 \\
    \hline
    6 & 7 &\nicefrac18&\nicefrac18&&&&&\nicefrac18&&\nicefrac18&&\nicefrac18&&&&&\nicefrac18&&&&&&&&\nicefrac14 \\
    \hline
    7 & 8 &\nicefrac1{12}&\nicefrac1{12}&&&\nicefrac1{12}&&&&&&\nicefrac16&&&&&\nicefrac1{12}&&&\nicefrac16&&\nicefrac16&&&\nicefrac16 \\
    \hline
    8 & 8 &\nicefrac1{12}&\nicefrac1{12}&&&&&\nicefrac16&&\nicefrac1{12}&\nicefrac1{12}&&&&&&\nicefrac16&\nicefrac16&&&&&&&\nicefrac16 \\
    \hline
    9 & 8 &\nicefrac16&\nicefrac1{12}&&&&&\nicefrac1{12}&&&\nicefrac1{12}&&\nicefrac16&&&\nicefrac1{12}&&&&\nicefrac16&&&&&\nicefrac16 \\
    \hline
    10 & 8 &\nicefrac18&&&&&\nicefrac18&&&&&&\nicefrac18&&&\nicefrac18&&&&\nicefrac18&\nicefrac18&\nicefrac18&&&\nicefrac18 \\
    \hline
    11 & 9 &\nicefrac1{12}&\nicefrac1{12}&&&\nicefrac1{12}&&\nicefrac1{12}&&&\nicefrac1{12}&\nicefrac1{12}&&&&&\nicefrac16&&&&&\nicefrac16&&&\nicefrac16 \\
    \hline
    12 & 9 &\nicefrac1{16}&\nicefrac1{16}&&&\nicefrac1{16}&&\nicefrac18&&&\nicefrac18&&&&&&\nicefrac3{16}&\nicefrac18&&&&\nicefrac18&&&\nicefrac18 \\
    \hline
    13 & 9 &\nicefrac1{12}&\nicefrac1{12}&&&\nicefrac1{12}&&&\nicefrac1{12}&&&\nicefrac1{12}&\nicefrac1{12}&&\nicefrac1{12}&\nicefrac1{12}&&&&\nicefrac13&&&&& \\
    \hline
    14 & 9 &\nicefrac1{10}&&\nicefrac1{10}&&&\nicefrac1{10}&&&&&&\nicefrac15&&\nicefrac1{10}&&&&&\nicefrac1{10}&\nicefrac1{10}&&&\nicefrac1{10}&\nicefrac1{10} \\
    \hline\hline
    \bf M &\bf \# & \bf E_0 & \bf E_1 & \bf E_2 & \bf E_3 & \bf E_4 & \bf E_5 & \bf E_6 & \bf E_7 & \bf E_8 & \bf E_9 & \bf E_{10} & \bf E_{11} & \bf E_{12} & \bf E_{13} & \bf E_{14} & \bf E_{15} & \bf E_{16} & \bf E_{17} & \bf E_{18} & \bf E_{19} & \bf E_{20} & \bf E_{21} & \bf E_{22} & \bf E_{23} \\
    \hline
  \end{tabular}
  \caption{Set of all extremal measurements (up to unitary transformations) for the bipartite HS Model, as given by Eq.~\eqref{eq:extmeas}. Measurements are labelled by $M$ and represented by a probability distribution $p$ over the set of extremal normalized effects $\{ E_j \}$. The number of non-null values of $p$ is also reported for convenience (in the column indicated by the symbol $\#$). We recall that factorized effects are those from $E_0$ to $E_{15}$ (included). Notice therefore that $M_1$ and $M_2$ are the only two instances with all effects factorized: they are also the only extremal measurements with extremal normalized effects possible in the PR Model (up to unitary transformations).}
  \label{tab:extmeas}
\end{table*}

The   set  of   extremal  measurements   with  extremal
normalized  effects   for  the  other  toy   models  is
necessarily    a   subset    of    those   listed    in
Table~\ref{tab:extmeas}.   For  example, the  PR  Model
allows  only  the  two  measurements  $M_1$  and  $M_2$
(again,  modulo  reversible transformations).   Indeed,
consistently with the  fact that the PR  model does not
contain entangled effects, $M_1$  and $M_2$ are the the
only instances whose  supporting normalized effects are
all factorized.

It is  relevant to observe  that the Hybrid  and Frozen
Models  include   $M_6$  or   one  of   its  equivalent
measurement  up  to reversible  transformations.   This
measurement  is   involved  in   the  no-hypersignaling
violation for the HS model  and by the same argument we
get the  same violation also  in the Hybrid  and Frozen
cases.
\end{document}